\documentclass[a4paper,UKenglish,cleveref, autoref, thm-restate,numberwithinsect]{lipics-v2021}
\usepackage{graphicx} 
\usepackage{macros}
\usepackage{caption}
\usepackage{tikz,pgfplots}
\usepackage{makecell,graphicx}

\hideLIPIcs
\nolinenumbers
\title{Submodular Max-Min Allocation under Identical Valuations}

\author{Kimon {Boehmer}}{DIENS, École normale supérieure, PSL University, Paris, France \and LIP, ENS Lyon, France  \and \url{https://www.kimonboehmer.net} }{kimon.bohmer@ens.fr}{0009-0002-5831-6158}{}

\authorrunning{K. Boehmer} 

\Copyright{Kimon Boehmer} 



\relatedversion{} 

\acknowledgements{I want to thank Chien-Chung Huang for introducing me to the topic and for many helpful dicussions.
Also, I thank the anonymous reviewers for their valuable suggestions.}

\EventEditors{Pierre Fraigniaud}
\EventNoEds{1}
\EventLongTitle{20th Scandinavian Symposium on Algorithm Theory (SWAT 2026)}
\EventShortTitle{SWAT 2026}
\EventAcronym{SWAT}
\EventYear{2026}
\EventDate{June 17--19, 2026}
\EventLocation{Copenhagen, Denmark}
\EventLogo{}
\SeriesVolume{370}
\ArticleNo{4}

\begin{document}

\maketitle

\begin{abstract}
    In the problem of \textsc{Submodular Max-Min Allocation}, we are given a set of items, a set of players, and monotone submodular valuation functions that represent the satisfaction of a player with a certain subset of items.
    The goal is to find an allocation of the items to the players that maximizes the lowest satisfaction among all players. 
    
    We study this problem in the special case where all players have the same valuation function.
    We devise a greedy algorithm which gives a $0.4$-approximation, improving the previously best factor of $\frac{10}{27} \approx 0.37$ by Uziahu and Feige.

    Furthermore, we study the integrality gap of the \emph{configuration LP} when players have identical valuations. By constructing a variable assignment to the dual from a primal integral solution, we give the first constant upper bound on the integrality gap for submodular valuations.
    Generalizing the result to the case where players' allocations must be independent in $k$ given matroids, we derive a $\OO(k)$-estimation algorithm for max-min allocation subject to $k$ matroid constraints under identical valuations.

\end{abstract}
\section{General context}

\textsc{Item allocation} is a fundamental problem in combinatorial optimization. Its task consists of distributing a set of \emph{items} to a set of \emph{players}, each equipped with a valuation function over subsets of items, with the goal of optimizing a specific objective function.
The most well-studied objective functions are \emph{max-sum}, i.e. maximizing the average satisfaction across all players, and \emph{max-min}, maximizing the satisfaction of the least satisfied player.
This work focuses on the latter objective function.
Applications of max-min allocation range from assigning classrooms to charter schools \cite{kurokawa2018leximin} to sensor placement and scheduling \cite{krause2009simultaneous}.

Usually, research focuses on restricted classes of valuation functions that are both expressive and computationally tractable.
For example, a common assumption is that each player has an \emph{additive} valuation function; that is, each item is associated with a certain weight, and the satisfaction of a player is the sum of the weights of the received items.
Another important and more general class of valuation functions are \emph{submodular} functions; 
besides their useful algorithmic properties, they naturally model the law of \emph{diminishing returns} in economics. This makes them especially interesting in the context of allocation.

In this work, we will restrict our attention to the case where all players share the same submodular valuation function.
It is interesting to note that the strongest hardness results of approximating submodular max-min allocation to a factor of $1-\frac{1}{e}$ already hold even for identical valuations~\cite{khot2008inapproximability,mirrokni2008tight}.

\subsection{Our results}
We design a simple greedy algorithm, which we call \emph{Truncated Max-Sum Greedy}, that achieves a $0.4$-approximation for identical submodular valuations, improving the previously best-known factor of $\frac{10}{27} \approx 0.37$ by Uziahu and Feige~\cite{uziahu2023fair}.
Essentially, given a fixed threshold $\alpha$, our algorithm greedily builds an allocation while ignoring players who already have items of value  $\alpha \cdot \OPT$. Our main result is the following:
 \begin{theorem}\label{th:2-5}
    There is a $\frac{2}{5}$-approximation for submodular max-min allocation under identical valuations.
 \end{theorem}
 While showing a guarantee of $\frac{1}{3}$ is rather straightforward, the improvement to $\frac{2}{5}$ uses some graph-theoretic tools such as Hall's theorem and certain arguments about the digraph of reallocations needed to transform the greedy solution into an optimal solution.
Our algorithm also provides guarantees in more general settings. As an example, we consider the problem of submodular max-min allocation under identical valuations with an additional \emph{global cardinality constraint}: The objective remains the same, but we are allowed to use at most $k$ items in total. We show the following:
    \begin{proposition}\label{prop:1-3-cardinality}
    There is a $\frac{1}{3}$-approximation for submodular max-min allocation under identical valuations with a global cardinality constraint.
    \end{proposition}
Furthermore, we study the integrality gap of the \emph{configuration LP} for max-min allocation, an important LP relaxation for max-min allocation.

We bound the integrality gap of the configuration LP by transforming certain primal integral solutions of the configuration LP to satisfying assignments of a dual formulation of this LP.
Unlike to similar approaches when rounding configuration LPs~\cite{svensson2011santa}, our analysis only uses the fact that the primal integral solution is \emph{at most} $\OPT$.
As a result, our proof does not provide an algorithm (not even a non-trivial exponential one) for constructing an integral solution that demonstrates the integrality gap. 

Using this technique, we give the first constant upper bound on the integrality gap of the configuration LP for identical submodular valuations.
\begin{restatable}{theorem}{theoremRestateThree} \label{th:3}
    The integrality gap of the configuration LP with identical submodular valuations is at most $3$.
\end{restatable}
Additionally, we give a lower bound on this integrality gap.
\begin{restatable}{proposition}{propLowerBoundSubmodularIntegralityGap}\label{prop:4-3}
    The integrality gap of the configuration LP with identical submodular valuations is at least $4/3$, even with $6$ items and $3$ players.
\end{restatable}
We believe that the technique of \cref{th:3} can be used to show constant integrality gaps for submodular max-min allocation in more general settings.
As an example, we also consider the case where we are given $k$ matroids on the items and the sets that players receive must all be independent in all matroids.
For this case, we also give a constant upper bound on the integrality gap, independent of $k$.
\begin{restatable}{theorem}{theoremRestateFive}\label{th:8}
    The integrality gap of the configuration LP with identical submodular valuations and any number of matroid constraints is at most $5$.
\end{restatable}
Accounting for the loss incurred when solving the separation oracle of the dual LP, this yields a $\OO(k)$-estimation algorithm for max-min allocation subject to $k$ matroid constraints.

\subsection{Related work}
\subparagraph*{Additive valuations.}
Additivity is one of the simplest assumptions on the valuation functions of the players.
In the general version where players may have non-identical additive valuation functions, the state-of-the-art algorithm by Chakrabarty, Chuzhoy and Khanna~\cite{chakrabarty2009allocating} achieves an approximation factor of $\OO(n^{\varepsilon})$. 
Interestingly, the best-known lower bound, shown by Bez{\'a}kov{\'a} and Dani~\cite{bezakova2005allocating}, only rules out a factor better than $2$.

A well-studied special case is the so-called \emph{restricted assignment} case, where players have identical valuations, but a bipartite graph between items and players restricts the allowed allocations.
To tackle this problem, Bansal and Sviridenko~\cite{bansal2006santa} introduced the configuration LP, and subsequent work \cite{bansal2006santa, asadpour2012santa, polavcek2015quasi, annamalai2017combinatorial} led to the current best approximation factor of $\frac{1}{4}-\varepsilon$ for the problem \cite{cheng2019restricted}.
Recently, the integrality gap of the configuration LP for restricted assignment was shown to be at most $3.53$~\cite{haxell2025improved}.
For identical valuations, Woeginger designed a polynomial-time approximation scheme (PTAS)~\cite{woeginger1997polynomial}.
\subparagraph*{Submodular valuations.}
The first approximation algorithm for general submodular valuations was given in \cite{golovin2005max}, and even though improvements have been made \cite{goemans2009approximating}, the best-known approximation factor remains polynomial in the input.
For a constant number of players, Chekuri, Vondrák and Zenklusen~\cite{chekuri2010dependent} gave a $1-1/e$-approximation.

For max-sum, Vondrak gave a $1-1/e$-approximation~\cite{vondrak2008optimal}, which was shown to be optimal~\cite{khot2008inapproximability,mirrokni2008tight}.
In the restricted assignment case, Bamas, Morell and Rohwedder~\cite{bamas2025submodular} generalized the techniques from Bansal and Sviridenko to obtain a $\mathcal{O}(\log \log n)$-approximation.
Under identical valuations, Krause, Rajagopal, Gupta and Guestrin~\cite{krause2009simultaneous} gave the first constant-factor approximation algorithm, which uses an algorithm for max-sum as a subroutine. Plugging in Vondrak's optimal algorithm, their method achieves a $\frac{e-1}{3e}\approx 0.21$-approximation.

More recently, another line of research has studied the problem of finding an allocation which approximates the \emph{maximin share} of all players, which is the largest value a player can ensure by proposing an allocation and receiving the bundle which is worst for them.
Under identical valuations, this problem is equivalent to submodular max-min allocation.
Barman and Krishnamurthy~\cite{barman2020approximation} gave a different algorithm with the same guarantee as Krause et al~\cite{krause2009simultaneous}.
Ghodsi, Hajiaghayi, Seddighin, Seddighin and Yami~\cite{ghodsi2018fair} improved this factor to $\frac{1}{3}$.
Finally, Uziahu and Feige~\cite{uziahu2023fair} gave a $\frac{10}{27}$-approximation by analyzing a bidding procedure which uses a truncated version of the valuation function.
\subparagraph*{Machine scheduling.}
If we reverse the roles of \emph{max} and \emph{min}, we obtain the classical \emph{makespan minimization} setting, where we try to assign jobs to machines so as to minimize the latest completion time.
As its max-min counterpart, it admits a PTAS for identical additive valuations~\cite{hochbaum1987using}. 
The submodular variant of makespan minimization, sometimes called \emph{submodular load balancing}, is not $o\left(\sqrt{\frac{n}{\log n}}\right)$-approximable~\cite{svitkina2011submodular}.
This stands in stark contrast to its ``dual'' max-min variant.

\section{Preliminaries}\label{sec:preliminaries}
Consider a set $J$ (with $n:=|J|$) of items and a \emph{valuation function} $f:2^J \rightarrow \RR$ that assigns a value to each subset of $J$.
We simplify notation by writing $f(j)$ instead of $f(\{j\})$, for $j \in J$.

We focus on the case where $f$ is \emph{submodular}: A function $f:2^J \rightarrow \RR$ is called submodular, if for every two subsets $S$ and $T$,
$$f(S \cup T) \leq f(S)+f(T) - f(S \cap T).$$
Closely related to submodular functions is the notion of \emph{marginal contribution}.
We define the marginal contribution on $f$ of an item $j$ with respect to a subset $S$ as $\Delta_f(j \mid S):=f(S \cup \{j\}) - f(S)$.
We extend this definition to sets (i.e. $\Delta_f(S \mid T)=f(S \cup T)-f(T)$).
In fact, a useful characterization of submodularity is that for all subsets $S \subseteq T$ of $J$ and $j \in J$, 
$$ \Delta_f(j \mid S) \geq \Delta_f(j \mid T).$$
A function is \emph{monotone} if for all $S \subseteq T$, we have $f(S) \leq f(T)$. We assume $f(\emptyset)=0$, and so a monotone function is non-negative.
In \ref{app:basic}, we give some basic facts about submodular functions that will be used implicitly throughout our proofs.

We define $P$ (with $m:=|P|$) to be the set of \emph{players}. 
Since we are focusing on identical valuations, each player $p$ has the same valuation function $f$ describing their satisfaction with each set of items.

More formally, a \emph{partial allocation} of $J$ is a partition $A=(A_1,...,A_m)$ of a subset of $J$.
Slightly abusing notation, we write $j \in A$ for $j \in \bigcup_{p \in P} A_p$.
If every item is allocated, we speak of an \emph{allocation}.
We refer to the sets of items in a (partial) allocation as \emph{configurations}.
We define $f_{\sf sum}(A)=\sum_{p \in P} f(A_p)$ and $f_{\sf min}(A):=\min_{p \in P}f(A_p)$.
We will refer to players $p$ with $f(A_p)=f_{\sf min}(A)$ as \emph{min players}.

We briefly define matroids since we discuss them in \cref{sec:4-3}. A \emph{matroid} $\mathcal{M} = (J,\mathcal{I})$ is a non-empty subset-closed set system in which for any $J' \subseteq J$, all sets of $\mathcal{I}$ restricted to $J'$ which are maximal have the same cardinality. The sets in $\mathcal{I}$ are called \emph{independent}.
\subsection{Submodular Max-Min Allocation}
With these definitions, we are ready to define the main computational problem of interest.
\problem{Submodular Max-Min Allocation}{A set of items $J$, a set of players $P=\{1,...,m\}$, and a monotone submodular function $f: 2^J \rightarrow \RR$.}{An allocation $A$ that maximizes $f_{\sf min}(A)$}

Notice that each player has the same valuation function $f$. We assume that we are given access to a so-called \emph{value oracle}, which, given any $S \subseteq J$, returns $f(S)$.

{\bf Approximation Hardness.}
Khot, Lipton, Markakis and Mehta~\cite{khot2008inapproximability} showed that when players have identical valuations, it is \compl{NP}-hard to distinguish between the case where all players can achieve value $1$ and the case where a player achieves at most value $1-\frac{1}{e}+\varepsilon$ on average.
Since the average value is always at least as big as the min value, their hardness result, while originally designed for max-sum, immediately translates to our setting:
\begin{corollary}[\cite{khot2008inapproximability}]
    Unless \emph{$\compl{P}=\compl{NP}$}, \emph{\textsc{Submodular Max-Min Allocation}} cannot be approximated within a factor $1-\frac{1}{e}+\varepsilon$ for any $\varepsilon>0$.
\end{corollary}

\section{Truncated Max-Sum Greedy}\label{sec:greedy}

In this section, we present our greedy algorithm for \textsc{Submodular Max-Min Allocation}.

Fix a threshold $\alpha \leq 1$.
We assume that the algorithm knows the optimal value $\OPT$. \footnote{We assume that the size of a rational representation of $f(S)$ is at most polynomial in $n+m$. We can therefore scale $f$ such that all values are integral and perform a binary search with start interval $[0,\overline{f}(J)]$, where $\overline{f}$ is the scaled version of $f$. Under our assumption, the binary search procedure will only recurse a polynomial number of times. }
First, the algorithm takes the $m$ items $j$ with largest $f(j)$-value and allocates one of them to each player.
Then, it greedily chooses the pair of available item $j$ and player $p$ that maximizes the marginal contribution $\Delta_{f}(j \mid A_p)$, under the constraint that $f(A_p) < \alpha \cdot \OPT$, i.e. we only consider players that have not reached an $\alpha$-fraction of their desired value.

The goal is to show that each player will eventually surpass this threshold, which proves an $\alpha$-approximation. 
The pseudocode is shown in \cref{algo:trunc}.
\begin{algorithm}
    \SetKwComment{Comment}{/* }{ */}
    \caption{The truncated max-sum greedy}\label{alg:greedy-maxsum}
    \KwData{items $J$, players $P$, submodular function $f: 2^J \rightarrow \RR$}
    \KwResult{an allocation $A$}
    Remove the $m$ items $j$ with largest $f(j)$ from $J$, and give one of them to each player\; 
   \While{$J \neq \emptyset$ and there exists $p \in P$ with $f(A_p) < \alpha \cdot \OPT$}
   {
    $(j^*,p^*) \gets \arg \max_{j \in J, p \in P: f(A_p)<\alpha \cdot \OPT} \Delta_f(j \mid A_p)$\;
    $J \gets J - \{j^*\}$\;
    $A_{p^*} \gets A_{p^*} \cup \{j^*\}$\;
   }
   \Return{$(A_1,...,A_m)$}
   \label{algo:trunc}
\end{algorithm}
\begin{observation}\label{obs:small-items}
    We can assume $f(j)<\alpha \cdot \OPT$ for all $j \in J$.
\end{observation}
\begin{proof}
    We show that if the truncated max-sum greedy gives an $\alpha$-approximation on instances where for all items $j \in J$, $f(j)<\alpha \cdot \OPT$, then it also gives an $\alpha$-approximation for the general case.

    Let $\OPT$ denote the optimal value in the general instance.
    Assume that there are $k$ \emph{big} items $j$ with $f(j)\geq \alpha \cdot \OPT$.
    In an optimal allocation, these items can be given to at most $k$ players.
    Thus, there exists an instance with $\max(0,m-k)$ players and a subset $J' \subseteq \{j \in J \mid f(j)<\alpha\cdot\OPT\}$ of items that achieve at least $\OPT$.

    If greedy is performed on the instance with $\max(0,m-k)$ players and without the big items, then it will obtain a solution of value at least $\alpha \cdot \OPT$ (since the optimal solution on this reduced instance has value $\OPT$, as argued above).
    Adding $k$ players, each containing one big item, gives a solution of value at least $\alpha \cdot \OPT$ to the general instance.
\end{proof}

As a warm-up, we first show that we obtain a $\frac{1}{3}$-approximation when setting $\alpha=\frac{1}{3}$.
Afterwards, we will show how to improve the factor to $\frac{2}{5}$ using a more refined analysis.

    Consider a greedy allocation $A$ and an optimal allocation $A^*$.
    Let $q$ be a min player in $A$, i.e. $f(A_q)=f_{\sf min}(A)$.
    We may assume $n \geq m$, since otherwise the optimal value is $0$.
    We first show the following simple but crucial lemma.
    \begin{lemma} \label{lem:help-max-sum-greedy}
        Assume $f(A_q)<\alpha \cdot \OPT$. Then,
        $$\sum_{j \in A^*} \Delta_f(j \mid A_q) \leq \sum_{p \in P-\{q\}}f(A_p)$$
    \end{lemma}

    \begin{proof}
        For any item $j$, let $A^{(j)}$ denote the partial allocation of the algorithm immediately before $j$ is allocated. 
        We have $\sum_{j \in A_q} \Delta_f(j \mid A_q)=0$. For all other $p \in P-\{q\}$, we have
        \begin{align}
            \sum_{j \in A_p} \Delta_f(j \mid A_q)
            \leq \sum_{j \in A_p}\Delta_f(j \mid A^{(j)}_q)
             \leq \sum_{j \in A_p} \Delta_f(j \mid A^{(j)}_p)
            = f(A_p)\label{eq:delta-smaller-p}
        \end{align}
        The first inequality follows from the submodularity of $f$, since $A^{(j)}_q \subseteq A_q$.
        For the second inequality, notice that the algorithm always chooses to allocate an item to the player $p$ with maximal marginal contribution, as long as $f(A_p)<\alpha \cdot \OPT$.
        Therefore, 
        \begin{align}
                    \sum_{j \in A^*} \Delta_f(j \mid A_q) = \sum_{p \in P-\{q\}} \sum_{j \in A_p} \Delta_f(j \mid A_q) + \sum_{j \in A_q} \Delta_f(j \mid A_q) \leq \sum_{p \in P-\{q\}} f(A_p).\label{eq:0}
        \end{align}
    \end{proof}
    We also need an upper bound on $f(A_p)$. The following lemma is sufficient for the $\frac{1}{3}$-approximation, but later we will need a stronger upper bound for the $\frac{2}{5}$-approximation.
    \begin{lemma}\label{lem:help2-max-sum-greedy}
        For all players $p \in P$, $f(A_p) \leq 2\cdot \alpha \cdot \OPT$.
    \end{lemma}
    \begin{proof}
        Player $p$ can only be allocated one item $j$ that surpasses $\alpha \cdot \OPT$. This item has size $ \leq \alpha \cdot \OPT$, by \cref{obs:small-items}.
        Let $A'$ be the partial allocation before $j$ is allocated to $p$. Then,
        $$f(A_p)=f(A'_p)+\Delta_f(j \mid A'_p)\leq \alpha\cdot\OPT + f(j) \leq 2\cdot\alpha \cdot \OPT$$
    \end{proof}
    We can now prove the approximation ratio of $\frac{1}{3}$ by the following lemma.
    \begin{lemma}\label{lem:1-3-approx}
        $f_{\sf min}(A^*) < f(A_q) + \frac{1}{m} \sum_{p \in P} f(A_p)$.
    \end{lemma}
    \begin{proof}
        \begin{align}
            f_{\sf min}(A^*) &\leq \min_{p \in P} f(A^*_p \cup A_q) \label{eq:1}\\
            & \leq \min_{p \in P} f(A_q) + \sum_{j \in A^*_p} \Delta_f(j \mid A_q)\nonumber\\
            & \leq f(A_q) + \frac{1}{m}\sum_{p \in P} \sum_{j \in A^*_p} \Delta_f(j \mid A_q)\nonumber\\
            & = f(A_q) + \frac{1}{m}\sum_{j \in A^*} \Delta_f(j \mid A_q)\nonumber\\
            & \leq f(A_q) + \frac{1}{m}\sum_{p \in P-\{q\}}f(A_p)\label{eq:4}\\
            & < f(A_q) + \frac{1}{m}\sum_{p \in P}f(A_p)\label{eq:5}
        \end{align}
    (\ref{eq:1}) uses monotonicity of $f$, (\ref{eq:4}) uses \cref{lem:help-max-sum-greedy} and (\ref{eq:5}) holds since $f(A_q)>0$ by our assumption $n \geq m$ and the fact that greedy starts by allocating one item to each player.
    \end{proof}
        To see that when setting $\alpha=\frac{1}{3}$, this indeed implies the approximation ratio of $\frac{1}{3}$, substitute $\frac{1}{m}\sum_{p \in P} f(A_p)\leq 2\cdot\alpha \cdot \OPT$ by \cref{lem:help2-max-sum-greedy}.
    Taking $2\cdot\alpha \cdot \OPT=\frac{2}{3}\cdot \OPT$ on the other side, we obtain $\frac{1}{3}\OPT < f(A_q)$.

\begin{remark}
    This algorithm extends to the case where we are additionally given a \emph{global cardinality constraint}. 
    That is, we must choose a subset of at most $k$ items to allocate.

    To see this, consider again an optimal allocation $A^*$ and the greedy allocation $A$. Notice that in our new setting, we might have $\bigcup_{p \in P} A_p \neq \bigcup_{p \in P} A^*_p$ and therefore the first equality in (\ref{eq:0}) no longer holds.
    However, we can circumvent this issue as follows. We can assume that both $A$ and $A^*$ contain exactly $k$ items by monotonicity.
    Let $b: \bigcup_{p \in P} A^*_p \rightarrow \bigcup_{p \in P} A_p$ be any bijection such that $b(j)=j$ for all $j \in (\bigcup_{p \in P} A^*_p) \cap (\bigcup_{p \in P} A_p)$.
    For all $j \in A$, let $A^{(j)}$ be the partial greedy allocation exactly before allocating $j$.
    Consider any $j \in A^*$ and let $b(j) \in A_p$. We have
    \begin{align}
        \Delta_f(j \mid A_q) \leq \Delta_f(j \mid A_q^{(b(j))}) \leq \Delta_f(b(j) \mid A_p^{(b(j))}) \label{eq:card-1}
    \end{align}
    
    using that when $b(j)$ is allocated, item $j$ is still available and thus greedy must prefer allocating $b(j)$ to $p$ than allocating $j$ to $q$.
    By (\ref{eq:card-1}) and the fact that $b$ is a bijection,
    \begin{align*}
            \sum_{j \in A^*} \Delta_f(j \mid A_q) &\leq \sum_{p \in P-\{q\}}\sum_{j \in A^*: b(j) \in A_p} \Delta_f(b(j) \mid A_p^{(b(j))}) \\
            &= \sum_{p \in P-\{q\}}\sum_{j \in A_p} \Delta_f(j \mid A_p^{(j)}) \\
            &= \sum_{p \in P-\{q\}} f(A_p)
    \end{align*}
    reproving \cref{lem:help-max-sum-greedy} in this more general setting. The rest of the analysis can be done as for the unconstrained case.
    Therefore, \cref{prop:1-3-cardinality} holds.
\end{remark}

\subsection{Ratio of 2/5}
We now prove that if $\alpha=\frac{2}{5}$, then $f_{\sf min}(A) \geq \frac{2}{5}\OPT$.
To this end, suppose for contradiction that there exists a counterexample to this claim.
We assume that this instance $(P,J,f)$ is minimal with respect to $|P|$ and that $f$ is scaled in such a way that $f_{\sf min}(A^*)=\OPT=5$.
We also fix a min player $q$ with $f(A_q) \leq 2$. Note that we could have assumed $f(A_q)<2$ but this is not necessary for the proof.

Let us define ${\sf avg}:=\frac{1}{m}\sum_{p \in P} f(A_p)$.
 From \cref{lem:1-3-approx} we obtain ${\sf avg} > 3$.
Our goal is to show that ${\sf avg} \leq 3$, thereby arriving at a contradiction.

Letting $a_p^i$ denote the $i$-th item allocated to $p$, define $\beta:=\max_{p \in P: |A_p| \geq 2} \Delta_f(a_p^2 \mid \{a_p^1\})$ as the largest marginal contribution of an item during the greedy algorithm which is not the first item of a player.
Notice that it suffices to consider $a_p^2$, since the marginal contributions of later items of a player can only be smaller.
  If $\beta \leq 1$, then no player can obtain more than value $3$, so we can assume $\beta>1$.
  Also, $\beta \leq 2$ holds by \cref{obs:small-items}.
  Notice finally that for $p \in P$, we have $f(a_p^1)\geq\beta$, as otherwise the maximizer $a_p^2$ in the definition of $\beta$ would have been among the largest $m$ items, and could not be the second item of a player.
\begin{lemma}\label{lem:bigger-subset}
    For each set of $k \geq 1$ players $p_1,...,p_k \in P-\{q\}$, there is no set of $k$ players $r_1,...,r_k \in P$ such that $\bigcup_{i=1}^k A_{p_i} \subseteq \bigcup_{i=1}^k A^*_{r_i}$.
\end{lemma}
\begin{proof}
    Removing $k$ players and all items in $\bigcup_{i=1}^k A_{p_i}$ from  the instance leads to a smaller counter-example:
    The greedy algorithm still behaves identically on all other players, including $q$, so there is a valid greedy allocation $A'$ on the instance with $f_{\sf min}(A') \leq 2$.
    On the other hand, since all removed items only belonged to $k$ players in the optimal solution, removing these $k$ players gives a solution $A^{**}$ with $f(A^{**}_p) \geq 5$ for all remaining players $p$.
\end{proof}

The above lemma allows us to prove the existence of a ``fixed'' item for each player:
\begin{lemma}\label{lem:hall}
    There is a permutation $\phi:P \rightarrow P$ such that $A_p \cap A^*_{\phi(p)} \neq \emptyset$ for all $p\in P$.
\end{lemma}
\begin{proof}
    Consider the bipartite graph $G=(P \cup P',E)$ with the set of players $P$ on the left side and a copy $P'$ of the set of players on the right side, and edges $(p,p')$ if $A_p \cap A^*_{p'} \neq \emptyset$.
    Let $N(Q) = \bigcup_{q \in Q} \{v \in V \mid (q,v) \in E\}$ be the neighbourhood of $Q$.
    Consider an arbitrary non-empty subset $Q \subseteq P-\{q\}$. By \cref{lem:bigger-subset}, $|N(Q)| \geq |Q| + 1$.
    For the set $\{q\}$, we have $|N(\{q\})| \geq 1$ since $A_q \neq \emptyset$ and each item in $q$ must be allocated somewhere in the optimal solution.
    For the set $\emptyset$, we have $|N(\emptyset)| \geq 0=|\emptyset|$.
    For any other set $\{q\} \subset Q \subseteq P$, we know that $|N(Q)| \geq |N(Q-\{q\})| \geq |Q-\{q\}|+1=|Q|$.
    Altogether, we obtain that for any subset $Q \subseteq P$, $|N(Q)| \geq |Q|$. We can hence apply Hall's Marriage Theorem to conclude that $G$ contains a perfect matching.
    This matching defines exactly the desired permutation $\phi$. 
\end{proof}
To simplify notation, we consider an optimal solution such that $\phi$ is the identity function, by just permuting the allocations of the players in the optimal solution.
We now define the \emph{reallocation digraph} $G=(P,E)$: Each pair of players $(p,r) \in P^2$ is associated to a set $A_{p \rightarrow r}=A_p \cap A^*_{r}$.
An arc $(p,r)$ exists if $A_{p \rightarrow r} \neq \emptyset$. By \cref{lem:hall}, each vertex has a self-loop.

We subdivide the players of our instance into two categories $P=B \cup S$.
The set $B$ of \emph{big} players consists of all players $b$ such that $f(A_b) > 4-\beta$. 
The set $S$ of \emph{small} players consists of all players $p$ with $f(A_p) \leq 4-\beta$.

Notice that $q \in S$.
This partition is mainly motivated by the following observation:
\begin{observation}\label{obs:exactly-two}
    For any $p \in B $, it holds that $|A_p|=2$.
\end{observation}
\begin{proof}
    Any player $p \in B $ cannot be allocated just one item $j$ since $f(j)<2 \leq 4-\beta \leq f(A_p)$.
    If $p$ was allocated $k$ items for $k>2$, then $f(\{a_p^1,\dots,a_p^{k-1}\})<2$ must hold in order for the greedy algorithm to allocate $a_p^k$ to $p$.
    But since $f(a_p^1) \geq \beta$, we must have $\Delta_f(a_p^k \mid \{a_p^1,\dots,a_p^{k-1}\}) \leq \Delta_f(a_p^2 \mid \{a_p^1\}) \leq 2 - \beta$.
    So we have $$f(A_p)=f(A_p-a_p^k) + \Delta_f(a_p^k \mid \{a_p^1,\dots,a_p^{k-1}\})< 2 + 2-\beta\leq f(A_p),$$ contradicting our assumption $k>2$.
\end{proof}

For each $p \in B$, denote by $x_p$ the item in $A_{p \rightarrow p}$ and denote by $y_p$ the other item of $A_p$.
\Cref{lem:hall} guarantees the existence of the first item and \cref{obs:exactly-two} implies that there is exactly one other item. \Cref{lem:bigger-subset} guarantees that the other item is not in $A_{p \rightarrow p}$.
\begin{figure}
    \usetikzlibrary{arrows.meta}

\begin{tikzpicture}[>=Stealth, node distance=3cm, every node/.style={circle, draw, minimum size=1cm}, auto]

  \node (A) at (0,0) {$b$};
  \node[color=gray] (B) at (3,0) {$c$};
  \node (C) at (6,0) {$p$};

   \draw[->] (A) edge[loop above] node[above=2pt,fill=white, draw=none, inner sep = 0pt, minimum size = 0pt, outer sep = 0pt] {$x_b$} (A);
  \draw[->, color=gray] (B) edge[loop above] node[above=2pt,fill=white, draw=none, inner sep = 0pt, minimum size = 0pt, outer sep = 0pt] {$x_c$} (B);

  \draw[->, color=gray] (A) edge node[above=2pt,fill=white, draw=none, inner sep = 0pt, minimum size = 0pt, outer sep = 0pt] {$y_b$} (B);
  \draw[->, color=gray] (B) edge node[above=2pt,fill=white, draw=none, inner sep = 0pt, minimum size = 0pt, outer sep = 0pt] {$y_c$} (C);

  \draw[->, dashed, red] (A) edge[bend left=60] node[above=2pt,fill=white, draw=none, inner sep = 0pt, minimum size = 0pt, outer sep = 0pt] {$z$} (C);

\end{tikzpicture}
    \centering
    \caption{Visualization of \cref{lem:no-edge}. }
\end{figure}
\begin{lemma}\label{lem:no-edge}
    For any  $b,c \in B$ with $b \neq c$, $(b,c)\not\in E$.
\end{lemma}
\begin{proof}
    Suppose for contradiction that $A_{b \rightarrow c} \neq \emptyset$, where $y_b \in A_{b \rightarrow c}$.
    Let $p \in P$ be the player such that $y_c \in A^*_p$.
    Consider a modified instance where we remove player $c$ as well as items $\{y_b,x_c,y_c\}$ and add an item $z$ which contributes an additive factor of $t:=\max(f(y_b),\Delta_f(y_c \mid A^*_p - \{y_c\}))$ to every set.
    By \cref{obs:restr} and \cref{lem:subm-add}, this function is submodular and monotone.
    Notice that this addition of $z$ does not change any marginal contribution, i.e. for $j \neq z$ we have $\Delta_f(j \mid S) = \Delta_f(j \mid S \cup \{z\})$.
    
    The greedy algorithm on the new instance may allocate all items that give marginal contribution larger than $t$ in the original run as before.
    As soon as all remaining items give marginal contribution $\leq t$, it may choose to allocate item $z$ to player $b$ (since $z$ contributes the same to each player, and greedy considers $b$ as $f(A_b)=f(\emptyset)<2$ or $f(A_b)=f(x_b)<2$).
    
    Since $b \in B$, we have $f(A_b)>4-\beta \geq 2$.
    Furthermore, $f(\{x_b,z\}) \geq f(x_b)+f(y_b) \geq f(A_b) \geq 2$.
    Consider the moment in the greedy algorithm after $z$ was allocated to $b$. If $x_b$ had already been allocated, then $b$ has reached the threshold of $2$ and will be thus ignored by the greedy algorithm. Therefore, all following items can be allocated to the same players as in the original run of the greedy algorithm.
    Now suppose $x_b$ had not yet been allocated. Since the addition of $z$ to player $b$ does not influence the marginal contributions and $f(z) <2$, the greedy algorithm can allocate all items up to $x_b$ in the same way (including $x_b$). After allocating $x_b$, player $b$ has reached the threshold of $2$ and thus again the greedy algorithm can allocate in the same way as in the original run.
    Thus, we again obtain $f_{\sf min}(A'_q)\leq 2$, for the new allocation $A'$.

    On the other hand, we can modify our old optimal allocation, by just replacing item $y_c$ by item $z$.
    Call the new optimal allocation $A^{**}$. We have
    \begin{align*}
        f(A^{**}_p)&=f(A^{**}_p - \{z\}) + t \\
        &= f(A^*_p - \{y_c\})+t \\
        &\geq f(A^*_p - \{y_c\})+ \Delta_f(y_c \mid A^*_p - \{y_c\})\\
        &=f(A^*_p) \geq 5
    \end{align*}
We obtained a counterexample with less players, a contradiction.
\end{proof}
Thus, in fact, $G[B]$ is an independent set.
We now argue that $B$ cannot be too large. 
To this end, we first prove the following two helpful lemmata.
For a player $p$, let us define $V_p:=\bigcup_{(r,p)\in E} A_{r \rightarrow p}$ the set of items reallocated \emph{to} $p$ and $W_p:=\bigcup_{(p,r)\in E} A_{p \rightarrow r}$ the set of items reallocated \emph{from} $p$.
\begin{lemma}\label{lem:give-at-most-two}
    For any player $p$, $\sum_{j \in W_p} \Delta_f(j \mid A_q)< 2$.
\end{lemma}
\begin{proof}
    Recall that $A^{(j)}$ denotes the partial greedy allocation exactly before item $j$ is allocated.
    Let $\ell$ be the last item allocated to $A_p$.
    By \cref{lem:hall}, there exists an item $k \in A_{p \rightarrow p}$.
    Therefore, $W_p \subseteq A_p - \{k\}$ and we can write
    \begin{align*}
        \sum_{j \in W_p} \Delta_f( j \mid A_q)
        &\leq \sum_{j \in A_p-\{k\}} \Delta_f(j \mid A_q)\\
        &\leq \sum_{j \in A_p-\{k\}} \Delta_f(j \mid A_q^{(j)})\\
        &\leq \Delta_f(\ell \mid A_p^{(\ell)}) + \sum_{j \in A_p-\{k,\ell\}} \Delta_f(j \mid A_p^{(j)})\\
        &\leq \Delta_f(\ell \mid A_p^{(k)}) + \sum_{j \in A_p-\{k,\ell\}}  \Delta_f(j \mid A_p^{(j)})\\
        &\leq \Delta_f(k \mid A_p^{(k)}) + \sum_{j \in A_p-\{k,\ell\}}  \Delta_f(j \mid A_p^{(j)})\\
        &=f(A_p-\{\ell\}) <2
    \end{align*}
    The third inequality is because the items were allocated to $p$ and not $q$.
    The fourth inequality comes from $A_p^{(\ell)} \supseteq A_p^{(k)}$, then we use that at the partial allocation $A^{(k)}$, greedy prefered $k$ over $\ell$. Finally, the strict inequality holds because otherwise $p$ would have achieved the $\alpha\cdot\OPT$ threshold and not received any more items.
\end{proof}

\begin{lemma}\label{lem:no-more-than-it-gets}
    For any $p \in P$, $\sum_{j \in V_p} \Delta_f(j \mid A_q) \geq 3+\sum_{j \in W_p} \Delta_f(j \mid A_q) - f(A_p)$.
\end{lemma}
\begin{proof}
    First, notice that
    \begin{align}
        \sum_{j \in A_p^*} \Delta_f(j \mid A_q) \geq \Delta_f(A_p^* \mid A_q) \geq f(A_p^*)-f(A_q) \geq 5-2 = 3. \label{eq:20}
    \end{align}
    Since $V_p = A_p^* - (A_p - W_p)$, $A_p-W_p \subseteq A_p^*$, and $W_p \subseteq A_p$, we can write
    $$\sum_{j \in V_p} \Delta_f(j \mid A_q) = \sum_{j \in A_p^*} \Delta_f(j \mid A_q)+\sum_{j \in W_p} \Delta_f(j \mid A_q)-\sum_{j \in A_p} \Delta_f(j \mid A_q)$$
    Now, the result follows by using (\ref{eq:20}) and (\ref{eq:delta-smaller-p}).
    \end{proof}
\begin{lemma}\label{lem:small-geq-big}
    $|B| \leq |S|$.
\end{lemma}
\begin{proof}
    Let $S_0$ be the set of small players that do not receive any item from big players, let $S_1$ be the set of small players that receive one item from one big player and let $S_{\geq 2}$ be the set of small players that receive items from multiple big players.
    If we can show $S_{\geq 2} = \emptyset$, this means that no two big players have an edge towards the same small player. But since every big player has an edge towards a small player by \cref{lem:no-edge}, by pigeonhole principle this would imply the desired $|B| \leq |S|$.
    
     Associate with each player $p$ the quantity $$g(p):=\sum_{j \in W_p} \Delta_f(j \mid A_q) - \sum_{j \in V_p} \Delta_f(j \mid A_q)$$
    Note that $\sum_{p \in P} g(p)=0$, as $\bigcup_{p \in P} W_p = \bigcup_{p \in P} V_p$.
    We can decompose the set of players as $$P=S_0 \cup \bigcup_{p \in S_1} \{p,b_1(p)\} \cup \bigcup_{p \in S_{\geq 2}} \{p,b_1(p),b_2(p),...,b_k(p)\}$$
    where $b_i(p)$ denotes the $i$-th big player that gives an item to $p$, in an arbitrary order. Therefore,
    \begin{align}
        \sum_{p \in P}g(p)=\sum_{p \in S_0} g(p) + \sum_{s \in S_1} g(s)+g(b_1(s)) + \sum_{s \in S_{\geq 2}}g(s)+g(b_1(s))+...+g(b_k(s))\label{eq:part}
    \end{align}
    Since $\sum_{p \in P} g(p)=0$, to prove $S_{\geq 2} = \emptyset$, it is enough to show that all terms of (\ref{eq:part}) are $\leq 0$, and that terms of the last sum are $<0$. This is the purpose of the rest of the proof.

    \begin{enumerate}
        \item $g(p) \leq 0$ for $p \in S_0$. This follows directly from \cref{lem:no-more-than-it-gets} and $f(A_p) \leq 4-\beta \leq 3$.
        \item $g(s) + g(b) \leq 0$ for $s \in S_1$ and $b = b_1(s)$.
        Since $s$ is a small player, we have $f(A_s) \leq 4-\beta$, and by \cref{lem:no-more-than-it-gets}, we get $g(s) \leq 1-\beta$.
    If we can show $g(b) \leq \beta-1$, then this implies $g(s)+g(b) \leq 0$.
    Recall that $W_b=\{y_b\}$. 

    As the second item allocated to $b$ contributes at most $\beta$, we have $\Delta_f(x_b \mid A_q) + \Delta_f(y_b \mid A_q) \leq 2 + \beta$. By the facts $f(A_b^*)-f(A_q) \geq 3$ and $A_b^*= V_b\cup(A_b-W_b)$, we get
    \begin{align*}
        \sum_{j \in V_b} \Delta_f(j \mid A_q)
        &\geq \Delta_f(V_b\cup(A_b-W_b) \mid (A_b-W_b) \cup A_q)\\
        &=\Delta_f(A_b^* \mid (A_b-W_b) \cup A_q)\\
        &\geq  f(A_b^*)-\Delta_f(A_b-W_b \mid A_q) - f(A_q)\\
        &\geq 3-\Delta_f(x_b \mid A_q)
    \end{align*}
In total, $g(b) \leq \Delta_f(y_b \mid A_q) - \sum_{j \in V_b} \Delta_f(j \mid A_q) \leq \Delta_f(y_b \mid A_q) - 3+\Delta_f(x_b \mid A_q) \leq \beta-1.$

        \item $g(s) + g(b) + g(c)<0$ for $s \in S_{\geq 2}$, $b=b_1(s)$ and $c=b_2(s)$.
         We have $f((A_b - W_b) \cup A_q) \leq f(A_b-W_b)+f(A_q) \leq 4$, but since $f(A_b^*)\geq 5$, we must have
          $$\sum_{j \in V_b} \Delta_f(j \mid A_q) \geq \Delta_f(V_b \cup (A_b-W_b) \mid (A_b-W_b) \cup A_q) \geq 1.$$
    The same holds for player $c$. By \cref{lem:give-at-most-two}, we also know $g(s)<2-\Delta_f(y_b \mid A_q)-\Delta_f(y_c \mid A_q)$.
    Therefore $$g(b)+g(c)+g(s) \leq \Delta_f(y_b \mid A_q)-1 + \Delta_f(y_c \mid A_q)-1+g(s)<0.$$
    \end{enumerate}
\end{proof}
We can conclude as follows:
\begin{align*}
{\sf avg}&=\frac{1}{m} \sum_{p \in P} f(A_p) \\
&\leq \frac{1}{m} \Bigl[|B| \cdot (2+\beta) + |S| \cdot (4-\beta) \Bigr]\\
&\leq \frac{1}{m} \Bigl[\frac{|B|+|S|}{2} \cdot (2+\beta)+ \frac{|B|+|S|}{2}\cdot (4-\beta)\Bigr]\\
&=\frac{1}{m}\Bigl[\frac{6 \cdot (|B|+|S|)}{2}\Bigr]\\
&=3
\end{align*}
where the last inequality is because $2+\beta \geq 4-\beta$, and $|B| \leq |S|$ by \cref{lem:small-geq-big}.
This contradicts ${\sf avg}>3$, as stated in the beginning. This finishes the proof of \cref{th:2-5}.

 We discuss the (close-to-) tightness of this analysis in \ref{app:A-tight}.
 In particular, we show that no value of $\alpha$ can lead to an approximation factor better than $\frac{2+S}{4+3S}\approx0.4068$, where $S = \sum_{i = 1}^\infty \frac{1}{s_i-1}$ with $s_k = 1+\prod_{i = 1}^n s_i$ is a constant related to the \emph{Sylvester sequence} or \emph{Salzer sequence} and was already used in \cite{liang1980lower,seiden2005two} for other packing problems. 

\section{Integrality gaps for the configuration LP}\label{sec:configuration-lp}
\begin{figure}[ht!]\small
\centering
\begin{minipage}{0.45\textwidth}
\captionsetup{type=figure}
\caption*{\textbf{Assignment LP (A)}}
\begin{align*}
    \sum_{i=1}^m x_{ij}&=1 & \forall j=1,...,n\\
    \sum_{j=1}^n f'_i(j)x_{ij} &\geq T & \forall i=1,...,m\\
    x_{ij} &\geq 0 & \forall i,j
\end{align*}
\end{minipage}
\hspace{1em}
\vrule width 1pt
\hspace{1em}\begin{minipage}{0.45\textwidth}
\captionsetup{type=figure}
\caption*{\textbf{Configuration LP (P)}}
\begin{align*}
    \sum_{C \in \mathcal{C}(T,p)} x_{p,C} &\geq 1 & \forall p \in P\\
    \sum_{C \ni j}\sum_{p \in P} x_{p,C} &\leq 1 & \forall j \in J\\
    x_{p,C} &\geq 0 & \forall p \in P, C \in \mathcal{C}(T,p)
\end{align*}
\end{minipage}
\caption{The assignment LP and the configuration LP.}
\label{LP:primals}
\end{figure}
\begin{figure}[ht!]\small
\centering
\begin{minipage}{0.33\textwidth}
\captionsetup{type=figure}
\caption*{\textbf{ Dual LP (D0)}}
\begin{align*}
    \max \sum_{p \in P} z_p &- \sum_{j \in J} y_j &  \text{s.t.}\\
    \sum_{j \in C} y_j &\geq z_p & \forall p,C\\
    y_j,z_p &\geq 0 & \forall j,p
\end{align*}
\end{minipage}
\vrule width 1pt
\begin{minipage}{0.33\textwidth}
\captionsetup{type=figure}
\caption*{\textbf{ Dual LP (D1)}}
\begin{align*}
    \sum_{p \in P} z_p &> \sum_{j \in J} y_j & \\
    \sum_{j \in C} y_j &\geq z_p & \forall p,C\\
    y_j,z_p &\geq 0 & \forall j,p
\end{align*}

\end{minipage}
\vrule width 1pt
\begin{minipage}{0.31\textwidth}
\captionsetup{type=figure}
\caption*{\textbf{ Dual LP (D2)}}

\begin{align*}
    \sum_{j \in J} y_j &<m &\\
    \sum_{j \in C} y_j &\geq 1 & \forall C\\
    y_j &\geq 0 & \forall j
\end{align*}
\end{minipage}
\caption{The duals to the configuration LP.}
\label{LP:duals}
\end{figure}
In this section, we study the integrality gap of the \emph{configuration LP} in the case of identical valuation functions.
Before defining this LP, let us consider the simpler \emph{assignment LP}:
It searches for the largest $T$ such that the program depicted in \cref{LP:primals} on the left, where $f'_i(j)=\min(f_i(j),T)$, has a feasible solution. Let us call this LP (A).
It is a folklore result that (A) has an integrality gap of $2$ when all players have identical additive valuations.

In the restricted assignment case, (A) has a bad integrality gap \cite{bansal2006santa}. This was the initial motivation to study a stronger LP, namely the aforementioned {configuration LP}.
Let $\mathcal{C}(T,p)$ contain the configurations $C$ such that $f_p(C) \geq T$.
We have a variable for each pair of player and configuration. The LP is shown in \cref{LP:primals} on the right.
The largest $T$ such that this LP has a feasible solution is an upper bound on the optimal solution. We call this system (P).

It is well-known that the integrality gap of (P) in the general case is $\Omega\bigl(\sqrt{m}\bigr)$ \cite{bansal2006santa}, while in the restricted assignment case it was shown to be smaller than $3.534$~\cite{haxell2025improved}.
When valuation functions are submodular, the gap in the restricted assignment case is $\mathcal{O}(\log \log n)$~\cite{bamas2025submodular}.

(P) has exponentially many variables in general. Assuming that the primal minimizes the $0$-vector, we can formulate its dual (D0), depicted in \cref{LP:duals} on the left.
The separation problem for this LP is the classical \textsc{Knapsack} problem. Since \textsc{Knapsack} admits a (F)PTAS~\cite{sahni1975approximate}, we can solve the LP up to any desired accuracy using the Ellipsoid method.
\subsection{Constructing Dual Solutions}
Notice that in (D0), a trivial solution always exists by setting all variables to $0$. But if we had a solution with objective value larger than $0$, then we could obtain an arbitrarily high objective value by scaling all variables.
By LP duality, we get that if (D0) is unbounded, then (P) with the same threshold $T$ is infeasible.
Additionally, we can directly express the requirement that the objective function must be greater than $0$ as a constraint. We call the resulting LP (D1), and it is shown in the middle of \cref{LP:duals}.

In the case of identical valuations, we can, without loss of generality, set all player variables to $1$, and obtain (D2), as seen on the right of \cref{LP:duals}.
This LP searches for a fractional hitting set of size smaller than $m$ that hits all configurations larger than our threshold value.
In (D2), we call the first constraint \emph{value constraint} and the second set of constraints \emph{configuration constraints}.
If (D2) has a feasible solution, then (D0) is unbounded, and by weak duality, (P) is infeasible. We obtain the following:
\begin{proposition}\label{prop:framework}
    Consider a class $\Gamma$ of max-min allocation instances, and let $\OPT(\mathcal{I})$ be the optimal integral value of instance $\mathcal{I}$.
    If for all $\mathcal{I} \in \Gamma$, (D2) is feasible for $\mathcal{I}$ and threshold value $\alpha \cdot \OPT(\mathcal{I})$, then the integrality gap of the configuration LP on $\Gamma$ is at most $\alpha$.
\end{proposition}

We give one more useful observation:
\begin{observation}\label{obs:lower}
    If there are variables $y_j$ such that the configuration constraints of (D2) are satisfied, $\sum_{j \in J} y_j \leq m$ and there exists $j^* \in J$ such that for all $C \ni j^*$, we have $\sum_{j \in C} y_j >1$, then (D2) is feasible.
\end{observation}
\begin{proof}
    Since $y_{j^*}$ is not involved in any tight constraint, we can lower its value by a sufficiently small amount, satisfying the value constraint while keeping all configuration constraints satisfied.
\end{proof}

\subsection{Submodular valuations}

The system (P) can also be formulated for submodular valuations. While in the additive case, solving the separation problem only incurred an $\varepsilon$ error, this is different in the submodular case. This is because we need to solve a submodular version of knapsack for the separation oracle.
Due to the hardness of maximum coverage~\cite{feige1998threshold}, which includes submodular knapsack, we lose a factor of $1-\frac{1}{e}$ while solving the LP. On the algorithmic side, this hardness is matched by a $1-\frac{1}{e}$-approximation by Sviridenko~\cite{sviridenko2004note}.
Therefore, any rounding algorithm that loses a factor of $\alpha$ translates into an approximation algorithm with factor $(1-\frac{1}{e})\alpha$.

In the following, we show that apart from our combinatorial approach presented in the last section, an LP rounding approach could also lead to a constant-factor approximation..
\theoremRestateThree*
We first show that we can restrict our attention to items $j$ with $f(j) \leq \OPT$. Let $\OPT(\mathcal{I})$ denote the optimal integral value for an instance $\mathcal{I}$, and let $\OPT_{\sf LP}(\mathcal{I})$ denote the optimal value of (P) on $\mathcal{I}$.
\begin{lemma}\label{lem:no-big-integral}
    Let $\mathcal{I}$ be an instance with an item $j$ such that $f(j) > \OPT(\mathcal{I})$.
    Let $\mathcal{I}'$ be the instance obtained from $\mathcal{I}$ by removing $j$ and one player. Then, $\OPT(\mathcal{I}) \geq \OPT(\mathcal{I}')$ and $\OPT_{\sf LP}(\mathcal{I}) \leq \OPT_{\sf LP}(\mathcal{I}')$.
\end{lemma}
\begin{proof}
    Fix an arbitrary player $p \in P$. For an optimal solution $x$ to (P), we have
    $$\sum_{C \ni j} \sum_{p \in P} x_{p,C} \leq 1 \leq \sum_{C \in \mathcal{C}(\OPT_{\sf LP}(\mathcal{I}),p)} x_{p,C}$$
    Therefore, we can replace all configurations in which $j$ appears by the conifigurations of player $p$, without violating any constraint except the constraint of player $p$. Since item $j$ is not used anymore, this is a valid solution to (P) on instance $\mathcal{I}'$ of value $\OPT_{\sf LP}(\mathcal{I})$.

    Now let $A$ be an optimal integral allocation of $\mathcal{I}'$ and assume for contradiction $f_{\sf min}(A)>\OPT(\mathcal{I})$. Adding a new player $r$ with $A_r = \{j\}$ gives an integral allocation for $ \mathcal{I}$ of value $\min(f_{\sf \min}(A), f(j)) > \OPT(\mathcal{I})$, a contradiction.
\end{proof}
Therefore, if our instance contains items that contribute more than $\OPT$, then we can instead consider a smaller instance with larger integrality gap.
We now prove \cref{th:3}.
\begin{proof}
    Consider an arbitrary instance $\mathcal{I}$ of our problem and scale the valuation function so that $\OPT(\mathcal{I})=1$.     
    We can assume that for all $j \in J$, $f(j)\leq 1$ by \cref{lem:no-big-integral}.
    Let $\overline{f}(S):=\min (2,f(S))$, and consider an allocation $A$ that maximizes $\sum_{p \in P} \overline{f}(A_p)$.

    For $j \in A_p$ and $f$, define $\Delta_f(-j \mid A_p):=\Delta_f(A_p \mid A_p-\{j\})$.
    For any $p \in P$, let $T_f(p):=\sum_{j \in A_p} \Delta_f(-j \mid A_p)$.
    For each $p \in P$ and $j \in A_p$, we set the variables in (D2) to
    $$y_j = \frac{\Delta_{\overline{f}}(-j \mid A_p)}{T_{\overline{f}}(p)}$$
   if $T_{\overline{f}}(p) \neq 0$, and $y_j=0$ otherwise. For notational convenience, define $y(S) = \sum_{j \in S} y_j$.
    We prove that $y$ is a (strictly) feasible solution to (D2) when $T=3$. Suppose for contradiction that there is a set $C$ of items with $f(C)>3$ but $y(C) \leq 1$.
    
    Let $q$ be the min player in $A$.
    Clearly, $f(A_q) \leq 1$, and thus $\Delta_f(C \mid A_q) > 2$.
    Let $j^*= \arg\max_{j \in C} \frac{\Delta_f(j \mid A_q)}{y_j}$. 
    If $\frac{\Delta_f(j^* \mid A_q)}{y_{j^*}} \leq 2$, then we would have
    $$\frac{\Delta_f(C \mid A_q)}{y(C)} \leq \frac{\sum_{j \in C} \Delta_f(j \mid A_q)}{y(C)} \leq \frac{\sum_{j \in C} 2 \cdot y_{j}}{y(C)} = 2$$
    which contradicts our assumption $y(C) \leq 1$. Therefore, we can assume $\frac{\Delta_f(j^* \mid A_q)}{y_{j^*}} > 2$.
    
    We want to argue that reallocating item $j^*$ to player $q$ increases our potential according to $\overline{f}$.
    Notice that the gain of adding $j^*$ to player $q$ is $\Delta_{\overline{f}}(j^* \mid A_q)$, while the loss of removing $j^*$ from player $p$ is $\Delta_{\overline{f}}(-j^* \mid A_p) = y_{j^*} \cdot T_{\overline{f}}(p)$.
    We now bound $T_{\overline{f}}(p)$. First, we have
     \begin{align*}
        T_{\overline{f}}(p)&=\sum_{j \in A_p} \Delta_{\overline{f}}(-j \mid A_p)\\
        &= \sum_{j \in A_p} \min\Bigl(0,\Delta_f(-j \mid A_p) - \max\Bigl(0,f(A_p)-2\Bigr)\Bigr)\\
        &= \sum_{j \in A_p : \Delta_f(-j\mid A_p) \geq f(A_p)-2} \Delta_f(-j \mid A_p) - \max\Bigl(0,f(A_p)-2\Bigr)\\
     \end{align*}
     If this sum is empty, then $T_{\overline{f}}(p)=0$. Otherwise, we can assume that $\max\bigl(0,f(A_p)-2\bigr)$ is subtracted at least once. Thus,
    \begin{align*}
        T_{\overline{f}}(p) & \leq \Bigl(\sum_{j \in A_p} \Delta_f(-j \mid A_p)\Bigr) - \max\Bigl(0,f(A_p)-2\Bigr)\\
        &\leq f(A_p) - \max\Bigl(0,f(A_p)-2\Bigr)\\
        &\leq 2
    \end{align*}
    where the second inequality is by submodularity of $f$.
    Therefore, $$\Delta_{\overline{f}}(j^* \mid A_q) - \Delta_{\overline{f}}(-j^* \mid A_p)>\Delta_{\overline{f}}(j^* \mid A_q) - 2 \cdot y_{j^*} = \Delta_{f}(j^* \mid A_q) - 2 \cdot y_{j^*}>0,$$
    where the equality is because of $f(A_q \cup \{j^*\}) = \overline{f}(A_q \cup \{j\})$, which follows from our assumption $f(j^*) \leq 1$ and $f(A_q) \leq 1$.
    So, reallocating $j^*$ from $p$ to $q$ increases the potential $\overline{f}$, contradicting the optimality of $A$.

    We conclude that the configuration constraints in (D2) are satisfied and not tight for $T>3$. We can finish the proof with \cref{obs:lower}, since
    $$\sum_{j \in J} y_j = \sum_{p \in P} \sum_{j \in A_p} y_j \leq m.$$
\end{proof}
A simple lower bound on the integrality gap is proven in \ref{app:integrality-gap}.
\propLowerBoundSubmodularIntegralityGap*
\subsection{Submodular valuations with matroid constraints}\label{sec:4-3}
Now consider the following problem: Apart from items $J$, players $P$ and a submodular monotone valuation function $f$, we are given $k$ matroids $\mathcal{M}_1=(J,\mathcal{I}_1),\dots,\mathcal{M}_k=(J,\mathcal{I}_k)$ and our task is to find a (partial) allocation $A$ such that $A_p \in \mathcal{I}_i$ for all $p \in P$ and all $i \in [k]$.
The configuration LP can be formulated in the same way as before, and the separation oracle incurs a loss of $2.7k$ by using the algorithm of \cite{chekuri2011submodular}.
In the following, we prove \cref{th:8}, which states that the integrality gap is still bounded by $5$ in this case and implies a $13.5k$-estimation algorithm for max-min allocation subject to $k$ matroid constraints.
\theoremRestateFive*
\begin{proof}
    As before, let $\OPT=1$, $\overline{f}(S):=\min (2,f(S))$ and define $A$ and the $y_j$-variables in the same way.
    The proof of \cref{lem:no-big-integral} directly translates to the matroid constraint setting, so we may again assume $f(j) \leq 1$ for all $j \in J$.
    Unlike before, we set $y_j=0$ for all $j \in A_q$, where $q$ is a fixed min player in $A$. We also set $y_j=0$ for all unallocated $j \notin A$.

    Now suppose for contradiction that there is a set $D \in \bigcap_{i = 1}^k \mathcal{I}_k$ of items with $f(D)>5$ but $y(D) < 1$.
    We can partition $D$ into $D_1 \uplus D_2$ such that $f(D_i) \geq 2$ for both $i \in [2]$, by greedily adding items to $D_1$ until $f(D_1)>2$. By the assumption that $f(j) \leq 1$ for all $j \in J$ and submodularity, this implies $f(D_2)=f(D-D_1)>2$.
    Now, choose the partition $D_i$ with the smaller $y$-value, and call it $C$.
    As a subset of $D$, clearly $C$ is independent in all matroids, and $y(C) < \frac{1}{2}$.

    Consider reallocating all items from $C$ to player $q$ and discarding all original items in $A_q$.
    We want to argue that this modified allocation has higher potential.
    The gain for player $q$ is $\overline{f}(C) - \overline{f}(A_q) \geq 2 - 1 = 1$.
    The loss for all other players is
    \begin{align*}
        \sum_{p \in P-\{q\}} \Delta_{\overline{f}}(-C \mid A_p) &\leq \sum_{p \in P-\{q\}} \sum_{j \in A_p} \Delta_{\overline{f}}(-j \mid A_p)\\
        &=\sum_{p \in P-\{q\}}\sum_{j \in A_p} y_j \cdot T_{\overline{f}}(p) \\
        &\leq y(C) \cdot 2\\
        &< 1 
    \end{align*}
since we can bound $T_{\overline{f}}(p)$ as in the proof of \cref{th:3}.
The total potential change is negative. This contradicts our choice of $A$.
Hence, the configuration constraints are all satisfied. The value constraint is satisfied since 
$$\sum_{j \in J} y_j = \sum_{j \notin A} y_j +\sum_{j \in A_q} y_j + \sum_{p \in P-\{q\}}\sum_{j \in A_p} y_j = \sum_{p \in P-\{q\}} \sum_{j \in A_p} y_j \leq m-1<m.$$
\end{proof}
Notice that we could replace matroids by any set system closed under subsets, but then of course we might not be able to solve the configuration LP efficiently.
\subsection*{Conclusion}
We presented an improved approximation algorithm for submodular max-min allocation under identical valuations, and showed some constant upper bounds on the integrality gap of the configuration LP under identical valuations.

However, several natural questions remain open. First, the proposed $0.4$-approximation is still far from the best-known complexity-theoretical hardness bound of $1-\frac{1}{e}\approx 0.63$. 
Another promising direction is to further explore the configuration LP. 
It would be valuable to obtain constructive, polynomial-time rounding algorithms that achieve the proven integrality gaps.
Furthermore, it would be interesting to see by how much we can relax the assumption of identical valuations. A natural first step in this direction would be to determine the integrality gap of the configuration LP in the case of submodular ``related machines'', where each player can be satisfied with a different value. 
\bibliography{bib}

\newpage

\appendix
\renewcommand{\thesection}{Appendix~\Alph{section}}
\renewcommand{\thelemma}{\Alph{section}.\arabic{lemma}}
\renewcommand{\theobservation}{\Alph{section}.\arabic{observation}}
\label[appsec]{appendix_definition}
\section{Basic Facts about Submodular Functions}\label{app:basic}
Let $f:2^J \rightarrow \RR$ be submodular and monotone.
\begin{lemma}\label{lem:subm-add}
    Consider a new item $z \notin J$.
    For any $t \in \RR$, define the function $f_t:2^{J\cup \{z\}} \rightarrow \RR$ by $f_t(S \cup \{z\})=f(S)+t$ and $f_t(S)=f(S)$ for all $S \not\ni z$.
    Then, $f_t$ is submodular and monotone.
\end{lemma}
\begin{proof}
    Monotonicity is clear. For submodularity,
    \begin{align*}
        f_t(S)+f_t(T) &= f(S)+f(T)+(\II_{z \in S} + \II_{z \in T}) \cdot t \\
        &\geq f(S \cup T)+f(S \cap T)+(\II_{z \in S} + \II_{z \in T}) \cdot t\\
        &= f(S \cup T)+f(S \cap T)+(\II_{z \in S \cup T} + \II_{z \in S \cap T}) \cdot t\\
        &=f_t(S \cup T)+f_t(S \cap T)
    \end{align*}
\end{proof}
The following statements are folklore. 
\begin{observation}\label{obs:restr}
    The restriction of $f$ to any subset $J' \subseteq J$ is submodular and monotone.
\end{observation}
\begin{observation}\label{obs:telescopic}
    Let $S=\{j_1,...,j_k\}$.
    Then, $\Delta_f(S \mid T)=\sum_{i=1}^k \Delta_f(j_i \mid T \cup \{j_1,...,j_{i-1}\})$.
\end{observation}
A direct corollary by using submodularity is:
\begin{observation}\label{obs:telescopic-submod}
    $\Delta_f(S \mid T) \leq \sum_{j \in S} \Delta_f(j \mid T)$.
\end{observation}
\begin{observation}\label{obs:marg}
    $f(S \cup T)=f(S) + \Delta_f(T \mid S)$.
\end{observation}
\begin{observation}\label{obs:zero-contribution}
    For $S \subseteq T$, $\Delta_f(S \mid T) = 0$.
\end{observation}
\section{Truncated Max-Sum Greedy} \label{app:A-tight}
In the following, we describe an instance where $\OPT=5$ but the greedy algorithm achieves only $2+\delta$ for a small $\delta$, independently of the choice of $\alpha$.
We first show that it is sufficient to find a certain bad example in an instance with an additive valuation function.
\begin{lemma}\label{lem:additive-app}
    Assume there is an instance $\mathcal{I}_{\sf{additive}}$ with an additive valuation function and a $\delta \geq 0$ with the following properties:
    \begin{itemize}
        \item There exists a partial allocation $A^*$ such that some item $j$ with $f(j)>0$ is \emph{not} allocated in $A^*$ and $f_{\sf min}(A^*) \geq 3-\delta$
        \item If $\alpha \cdot \OPT> 2 + \delta$, greedy can produce an allocation $A$ where all items are allocated.
    \end{itemize}
    Then, there exists an instance $\mathcal{I}_{\sf{submodular}}$ with a submodular valuation function where the optimal solution achieves at least $5$ and independently of the choice of $\alpha$, greedy can produce an allocation $A$ with $f_{\sf min}(A) \leq \frac{2+\delta}{5}$.
\end{lemma}
\begin{proof}
    To create $\mathcal{I}_{\sf{submodular}}$, start by copying $\mathcal{I}_{\sf{additive}}$ $\lceil \frac{5}{f(j)} \rceil$ many times (i.e. create new items and new players, while maintaining the additive weights of all items). Let $P'$ be the players and $J'$ be the items of the resulting instance.
    A valid run of greedy simulates each step of the run on $\mathcal{I}_{\sf{additive}}$ on every copy in ``parallel''.
    Therefore, the resulting allocation $A'$ still allocates all items.
    Clearly, there still exists an allocation $A'^*$ with $f_{\sf min}(A'^*) \geq 5$.

    Let $s=\min_{j' \in J'} f(j')$ (we can assume $s>0$, otherwise we could have discarded the item $j'$ with $f(j') = 0$), and let $0 < s' \leq s$ divide $2+\delta$.
    Create a new item $z^*$, and for $p \in P'$ and $i \in [N:=\frac{2+\delta}{s'}]$, create an item $z^p_i$.
    
    Let us define a new valuation function $f'$ on all items. 
    Set $\Delta_{f'}(z^p_i \mid S) = s'$ if $z^{p'}_i,z^* \notin S$ for all $p'$, and set $\Delta_{f'}(z^p_i \mid S) = 0$ otherwise.
    Set $\Delta_{f'}(z^* \mid S) = 2 + \delta - |\{j \mid \exists p: z^p_i \in S\}|$.
    For all $j' \in J'$, set $\Delta_{f'}(j' \mid S) = f(j')$ for all sets $S$.
    It is straightforward to check that $f'$ is a coverage function and thus submodular.

    Our instance $\mathcal{I}_{\sf{submodular}}$ consists of $J'$ together with $z^*$ and all $z^p_i$, valuation function $f'$ and players $P' \cup \{p^*\}$ for a single new player $p^*$.
    First, let us show that there exists an allocation $B^{*}$ with $f_{\sf min}(B^{*}) \geq 5$. For each player $p \in P'$, set $B^{*}_p = A'^*_p \cup \{z^p_j \mid j \in [N]\}$. Then, $f(B^{*}_p) = f(A'^*_p) + N \cdot s' \geq 3-\delta + 2+\delta = 5$.
    To the special player $p^*$, we give all copies of item $j$. Since there are $\lceil \frac{5}{f(j)} \rceil$ many copies, we get $f(B^{*}_p) \geq 5$.

    Now, let us argue that greedy can produce an allocation $B$ with $f_{\sf min}(B) \leq 2 + \delta$.
    We can assume $\alpha \cdot \OPT \geq 2+\delta+\varepsilon$ for some $\varepsilon >0$, since otherwise it is easy to construct an instance where greedy achieves at most $2+\delta$ (e.g., many items of additive value $2+\delta$).
    First, notice that greedy can start by allocating all additive items as in $A'$, since the $z^p_i$-items have smaller marginal contributions.
    At this point, $p^*$ has not received any item. Since for the remaining items $f(z^* \cup \bigcup_{p,j} z^p_i) = 2 + \delta$, in the final allocation $B$, player $p^*$ can achieve at most value $2 + \delta$. This finishes the proof.

\end{proof}
Now we are left with exhibiting an instance that satisfies the conditions of the lemma for a $\delta$ as small as possible.
Fix some $N \in \NN^+$ and let $s_1,\dots,s_N$ be the first $N$ members of the Sylvester sequence $s_n = 1 + \prod_{i = 1}^{n-1} s_i$. Let $m=2\cdot (s_N-1)$ be the number of players and $J$ contain the following items:
\begin{itemize}
    \item $m/2$ items of size $2+\delta$
    \item $m$ items of size $\frac{3-\delta}{2}$
    \item $m/2$ items of size $\frac{1+3\delta}{2(s_i-1)}$ for each $i \in [N-1]$
    \item $m/2+1$ items of size $\frac{1+3\delta}{2(s_N-1)}$. 
\end{itemize}
The following is a valid greedy allocation for threshold value $\alpha \cdot \OPT > 2+\delta$.
\begin{itemize}
    \item $m/2$ players receive $2+\delta$ and $\frac{3-\delta}{2}$
    \item For each $i \in [N-1]$, $\frac{m}{2s_i}$ players receive $\frac{3-\delta}{2}$ and $s_i$ items of size $\frac{1+3\delta}{2(s_i-1)}$ 
    \item one player receives $\frac{3-\delta}{2}$ and $s_N$ items of size $\frac{1+3\delta}{2(s_N-1)}$.
\end{itemize}
The number of items of size $\frac{3-\delta}{2}$ that are used is
$$\frac{m}{2} + \Bigl(\sum_{i = 1}^{N-1}  \frac{m}{2s_i}\Bigr) + 1 = \frac{m}{2} + \frac{m}{2} \cdot \Bigl(\sum_{i = 1}^{N-1} \frac{1}{s_i}\Bigr) + 1= \frac{m}{2} + \frac{m}{2} \cdot \Bigl(1-\frac{1}{s_N-1}\Bigr) +1= m$$
where the equality $\sum_{i = 1}^{N-1} \frac{1}{s_i} = 1-\frac{1}{s_N-1}$ is known and easy to show.
One can check that each player gets one of the $m$ largest items, that all items got allocated and that if we remove one item from a player, the sum of the remaining item sizes is at most $2+\delta < \alpha \cdot \OPT$.
The following is another solution:
\begin{itemize}
    \item $m/2$ players receive $2$ items of size $\frac{3-\delta}{2}$
    \item $m/2$ players receive $2+\delta$ and one item of size $\frac{1+3\delta}{2(s_i-1)}$ for each $i \in [N]$.
\end{itemize}
Notice that one item of size $\frac{1+3\delta}{2(s_N-1)}$ is not allocated. The first half of the players achieves value $3-\delta$.
For the second half to also achieve the same value, we need
$$2+\delta + \frac{1+3\delta}{2}\cdot S = 3-\delta$$
where $S = \sum_{i=1}^N \frac{1}{s_i -1}$. Therefore we can apply \cref{lem:additive-app} with
$$\delta = \frac{2-S}{4+3S}$$
When $N$ goes to infinity, $S \approx 1.691$. Thus the ratio between the greedy solution and the optimal solution is at least 
$$\frac{2+\delta}{5} = \frac{2+S}{4+3S} \approx 0.4068.$$
We believe that this is the worst possible ratio achievable by the greedy algorithm.
\section{Lower bound on the integrality gap}\label{app:integrality-gap}
\propLowerBoundSubmodularIntegralityGap*
\begin{proof}
Consider items $j_1,j_2,j_3,k_1,k_2,k_3$ and a submodular function defined as $f(j_a)=f(k_a)=2$ for all $a$, $f(\{j_a,j_b\})=f(\{k_a,k_b\})=4$ for all $a \neq b$, $f(\{j_a,k_b\})=3$ and $f(S)=4$ for all $S$ such that $|S| \geq 3$.
Our instance has $3$ players.
For all $j \in J$, we have $\Delta_f(j \mid S)=2$ when $|S|=0$, $\Delta_f(j \mid S) \in [1,2]$ when $|S|=1$, $\Delta_f(j \mid S) \in [0,1]$ when $|S|=2$ and $\Delta_f(j \mid S)=0$ when $|S|>2$, which implies submodularity of $f$ (monotonicity is clear).

The configuration LP is feasible for this instance with threshold value $4$: We set $x_{p,\{j_a,j_b\}}=x_{p,\{k_a,k_b\}}=1/6$ for all $p$ and $a \neq b$.
Each player gets one unit of configuration (because there are six configurations in the support) and each item appears only in two configurations of the support, and thus in at most one unit of configuration.

Now, assume that there is an integral solution of value at least $4$. 
Since there are six items and three players, and a singleton item only gets value $2$, each player must be allocated exactly $2$ items.
But in any allocation of configurations of two items to three players, one player must obtain items of the form $\{j_a,k_b\}$, which only gives value $3$.
\end{proof}
\end{document}